\documentclass[12pt]{amsart}
\usepackage{amsmath}
\usepackage{amssymb}
\usepackage{amscd}
\usepackage{pstricks}
\usepackage{color}
\usepackage{mathpazo}
\usepackage[dvips]{graphicx}
\usepackage{textcomp}

\theoremstyle{plain}

\newtheorem{theorem}{Theorem}
\newtheorem{lemma}{Lemma}
\newtheorem{definition}{Definition}

\newtheorem{corollary}[lemma]{Corollary}
\newtheorem{remark}{Remark}
\newtheorem{Fact}{Fact}

\def\reals{{\mathbb{R}}}
\def\integers{{\mathbb{Z}}}

\def\EXP{{\mathbb{E}}}

\def\Prob{{\mathbb{P}}}

\begin{document}

\title[Energy transfer and joint diffusion]
{Energy transfer and joint diffusion}

\author{Zsolt Pajor-Gyulai, Domokos Sz\'asz}
\

\thanks{D. Sz. is grateful to
Hungarian National Foundation for Scientific Research grants No. T 046187, K 71693,
NK 63066 and TS 049835}

\address{Zsolt Pajor-Gyulai: Budapest University of Technology\\   Institute of Physics
\\  Domokos Sz\'asz: Budapest University of Technology\\   Mathematical Institute
 \\ Budapest, Egry J. u. 1 Hungary H-1111\\
}
\email{pgyzs@math.bme.hu, szasz@math.bme.hu}
\maketitle

\vskip1cm
\begin{abstract}
A paradigm model is suggested for describing the diffusive limit of trajectories of two Lorentz disks moving in a finite horizon periodic configuration of smooth, strictly convex scatterers and interacting with each other via elastic collisions. For this model the diffusive limit of the two trajectories is a mixture of joint Gaussian laws (analogous behavior is expected for the mechanical model of two Lorentz disks).
\end{abstract}

\tableofcontents
{\bf Note:} MRP stands for Markovian Renewal Process, STRP for Scaled Type Renewal Process and STMRP for Scaled Type Markov Renewal Process (cf. section 2).\\

\today
\section{Introduction}
Beside the dynamics itself, the joint motion of two particles interacting with each other and with a dynamical environment also depends  on the spatial dimension. The first model where this question was addressed (cf. \cite{Sz80}) was a one-dimensional mechanical one. There, the asymptotically diffusive motions of the two particles were either glued together or were independent depending on the initial distance of the particles. The model was actually that of Harris and Spitzer, (see \cite{S69}) (equilibrium dynamics of elastically colliding point particles) generalized by Major and Sz\'asz, \cite{MSz80} (non-equilibrium dynamics). In a related one-dimensional - random - collision system (cf. \cite{KLPS83}) the joint motions were dependent. On the other hand,
Kipnis and Varadhan (\cite{KV86}) have shown that the diffusive limits of two particles in a symmetric exclusion process on $\mathbb Z^d$ with $d \ge 1$ are independent Brownian motions except the one-dimensional nearest neighbor case when the motions are subdiffusive.

It is worth noting that the joint motion of two particles got also studied in the physics literature, e. g. the mutual dynamics of pairs of atoms in a dense Lennard-Jones liquid in \cite{PBV84}.

Returning from stochastic dynamics to a deterministic one, let us consider the planar, finite-horizon Lorentz process with a periodic configuration of scatterers. It is known that its limit in the diffusive scaling is a Brownian motion (cf. \cite{BS81} and \cite{BCS91}). Of course, two point like Lorentz particles do not interact, but if we take two small disks then the case is quite different.

A simple fact: The motion of one small disk is itself isomorphic to a Lorentz process, so its diffusive limit is again  the Wiener process. However, if one considers two small Lorentz disks, then the na\"ive heuristics would suggest that, since the two particles collide very rarely (i. e. $O(\log n)$ times during the first $n$ collisions), the situation is similar to the locally perturbed Lorentz process where the diffusive limit is the same Brownian motion as it was for the unperturbed Lorentz process (cf. \cite{DSzV09}). This analogy is, however, misleading and the aim of the present work is exactly to clarify the situation. The difference with the preceding models is {\it the interaction: elastic collision of the disks also changes the energies of the two particles}. Moreover, in dimension two, by borrowing heuristics from random walk theory (cf. \cite{S76}) and estimates from \cite{DSzV08}, one can convince himself/herself that the time intervals between consecutive collisions have a slowly varying tail. Consequently, for large $t$, the last collision of the disks preceding $t$ befell at time $o(t)$ with a probability close to one. Thus the energies of the disks at time $t$ determining the limiting variance are the random energies obtained at the aforementioned last collision before $t$, ergo the diffusive limit of each disk is a Brownian motion with a random covariance (and their joint limit can already be calculated based upon the previous line of ideas).

The goal of the present work is to make the above heuristic argument precise on the level of a stochastic model mimicking the deterministic model of two Lorentz disks.

Our model is, roughly speaking, a colliding system of two random walks with internal states where the speeds of the particles are represented by exponential clocks and are included in the set of internal states. The model, our main conditions and our main result are described in subsections 2.1-2.2. Subsection 2.3 contains the main, often new, probablilistic concepts and results and a sketch of our proof. Section 3 is devoted to the verification of our local limit theorem for general random walks with internal states and some corollaries, whereas section 4 the the proof of our main Theorem 1. Section 5 contains some remarks. Finally, the proofs for our results for Markovian renewal processes and scaled type Markovian renewal processes is provided in the appendices.

\section{The Model. Main result. Methods}

\subsection{The Model}

The dynamics of two Lorentz disks will be modeled by two continuous random walks with continuous internal states whose steps are independent whenever the walkers are at different lattice sites. If they are on the same site, then their interaction is given by a collision operator (see below).

\subsubsection{Continuous Time Random Walks with Continuous Internal States} \hfill

Discrete time random walks with a finite number of internal states were introduced by Sinai, \cite{S81}  where the internal states were meant to represent elements of a Markov partition. The theory was elaborated in a series of works \cite{KSz83,KSz84,KSSz86}. In our case the internal states will also represent particle velocity therefore we have to consider random walks with internal states where the internal states belong to a more general state space. Moreover, for being able to include speed we take continuous time. In \cite{KSz83}, a local limit theorem was established for  random walks on $\integers^d$ with a finite number of internal states and we will also use much of the techniques presented there.

\begin{definition} (Sinai, 1980)\label{def:Sinai}
Let $H, |H| < \infty$ be the set of states. On the set ${\mathbb Z^d} \times {H}$ the Markov chain $\xi_n = ({\eta_n}, {\varepsilon_n})$ is a Random Walk with Internal States (RWwIS) if for $\forall \ x_n, x_{n+1}\  \in \mathbb Z^d, \ u_n \in H, A \subset H$
\[
P ( \xi_{n+1} = (x_{n+1}, u_{n+1}), u_{n+1} \in A  | \xi_n = (x_n, u_n) )= p_{{x_{n+1}-x_n}} {(u_n, A)}
\]
Of course, $\{{\varepsilon_n}; n \ge 0\}_n$ is also a Markov chain due to the spatial translation invariance.
\end{definition}

Our paradigm for the mechanical model will be introduced in two steps. First, the individual motion of each of the two particles will be a continuous time Markovian random walk with internal states (we will abbreviate it by RWwIS again; it will always be obvious whether we are talking about the discrete or the continuous time case) with some general state space $\tilde H$ and a constant exponential-jump rate $\lambda >0 $. (So far we do not specify $\tilde H$). We just note, however, that later $\lambda $ will be included among the internal states of the full two particle system to permit its change at collisions of two particles.

\begin{definition}\label{speed_param} Assume we are given a rate $\lambda >0$ and a family
\begin{equation}\label{subst_kernels}
\{P_x(v, .)| x \in \integers^d\setminus \{0\}\}
\end{equation}
of substochastic kernels over $\tilde H$ such that ${Q} = \sum_{x \in \mathbb Z^d\setminus \{0\}} P_x$ is a stochastic kernel over $\tilde H$.
A continuous time pure jump Markov process $\xi_t=(\eta_t, \varepsilon_t)$ with state space $\mathbb Z^d \times\tilde{H}$ - is called a (generalized) Random Walk with Internal States (RWwIS) if
\[
\Prob (\xi_{t+dt} = \xi_t) = 1 -\lambda dt + o(dt)
\]
and for every $(x_t, u) \in \mathbb Z^d \times \tilde H$, $\forall A \subset \tilde{H}$ and $\ x_{t+dt}-x_t\neq 0$
\[
\Prob (\xi_{t+dt} = (x_{t+dt}, v'), v' \in A| \xi_t = (x_t, v)) = \lambda P_{x_{t+dt}-x_t}(v, A) dt +o(dt)
\]
\end{definition}

In other words, the kernel for a jump to $x \in \mathbb Z^d\setminus \{0\}$ is described  by
\[
{P_x f (v)} = \int_{\tilde{H}} f(v') P_x(v, dv'), \qquad f \in L_\infty(\tilde{H})
\]
and the  transition operator for the discrete time Markov chain $\{\varepsilon_n\}_{n \ge 0}$ of subsequent internal states is $Q: L_\infty(\tilde{H}) \to L_\infty(\tilde{H})$.

From now on we will mainly restrict our discussion to the planar case (though we will briefly mention other cases, too).

As said, our RWwIS is to mimic Lorentz disk process in $\reals^2$. Since in the two particle process the energy, i.e. the rate of the particle will also change, it is appropriate to include this rate among the internal states. Concretely, we will have $\tilde{H} = S \times \mathcal{I}$ for the set of internal states which now also includes the rate $\lambda$. Here $S=\mathbb R / \mathbb Z$ stands for the direction $u=\frac{v}{|v|}$ of the  velocity of a particle and $\mathcal{I}$  for its speed $\lambda=|v|$. (Here $\lambda \in \mathcal{I} =  [a, b], \ \  0\leq a < b < \infty$; $\lambda$ will be, of course, conserved in the absence of interaction).

Without interaction, the generator for a single random walker can be decomposed
\[
Q=Q_S\otimes id_{\mathcal{I}}
\]
Also $P_x((u,\lambda),.)=P_x^S(u,.)\otimes \delta_{\lambda}$. Thus $Q_S=\sum_{x\in\mathbb{Z}^d\setminus \{0\}}P_x$ is indeed the stochastic kernel on $S$.

\subsubsection{Interaction: the collision operator}\label{sect:Interaction}\hfill

Next we define the collision interaction. {Let $\xi^i_t = (\eta^i_t, \varepsilon^i_t), \ i=1, 2$ be two RWwIS.

Whenever {$\eta^1_t \neq \eta^2_t $}, the joint generator of the two Markov processes is the product of the two individual generators modeling two independent Lorentz processes. On the other hand, when {$\eta^1_t = \eta^2_t\ \ (= x)$}, then

\begin{align*}
\Prob {\bf(}  \xi^1_{t+} = ( x+z^1, v^1_+), \xi^2_{t+} = &(x + z^2, v^2_+); v^1_+ \in A^1, v^2_+ \in A^2\\
&| \xi^1_{t-} = (x, v_-^1), \xi^2_{t-} = (x, v_{t-}^2){\bf)} \\
&=C_{z^1, z^2}{\bf(} v_-^1, v_-^2, A^1, A^2{\bf)}
\end{align*}
is the {collision kernel}. We assume that $C$ satisfies conservation of energy: $(v_-^1)^2+(v_-^2)^2 = (v^1_+)^2+(v^2_+)^2$ (momentum is not conserved since the collision kernel contains averaging over normal of impact, see below). Thus
\[
C_{z^1,z^2}(v_-^1,v^2_-,.,.)=C_{z^1,z^2}(\lambda^1_-,u_-^1,u_-^2,.,.,.)
\]
where $\lambda_-^1$ is the precollisional speed parameter of the first random walker (that of the other one is determined by energy conservation). For convenience, we will always use the speed of the first walker to describe the energy partition between the two particles.

We can and do assume that $(v^1)^2+(v^2)^2 = 1$. Therefore  the state space of the two particle process is isomorphic to $\left(\mathbb{Z}^2\times S\right)^2\times\tilde{\mathcal{I}}$ where $\tilde{\mathcal{I}}=[0,1]$ (in what follows we always assume $2E=1$). It is worth noting that the concrete form of the collision kernel for the mesoscopic version of two disk model is calculated in Appendix A of \cite{GG08}.\vspace{1mm}
\\
 {{\it Warning}: the pair $(\xi_t^1, \xi^2_t)$ is not a RWwIS on $\mathbb Z^d \times \mathbb Z^d$ anymore since translation invariance is hurt on the diagonal.}
\\

\subsubsection{Molecular chaos}
\begin{enumerate}
\item If we recall that our model is to mimic the two disk process, we note that the deterministic law driving the collision does not only involve $v_-^1$ and $v_-^2$, but also an angle describing the positions of the two disks relative to each other. As frequently in the literature, we assume that the distribution of this angle is uniform, and averaging over it gives our stochastic collision operator defined above.

\item As we will see later (cf. Theorem \ref{Haeusler-Mason}), short inter-collision times are extremely rare asymptotically, so the joint law of the directions of incoming velocities will approach an equilibrium distribution. The particular form of this law is not important, the point is that by averaging over $u_-^1,u_-^2$, we will use the mesoscopic collision kernel
\[
\tilde{C}_{z^1,z^2}(\lambda_-^1\to\lambda_+^1,u_+^1,u_+^2)
\]
This shorthand notation means that we will only use $\lambda_-^1$ from the precollisional data to compute the postcollisional velocities.
\end{enumerate}

\subsubsection{Summary of the model}

Finally, we have as our object of investigation
\[
J_t=(\xi_t^1,\xi_t^2,\Lambda_t)=(\eta_t^1,\varepsilon_t^1,\eta_t^2,\varepsilon_t^2,\Lambda_t)
\]

Introduce the notation $\tilde{\lambda}=\lambda+\sqrt{1-\lambda^2}$, i.e. the sum of the two rates and let $\lambda_i=\lambda$ if $i=1$ and $\lambda_i=\sqrt{1-\lambda^2}$ if $i=2$. Then the dynamics can be summarized as the following. First,
\[
\Prob(J_{t+dt}=J_t|\Lambda_t=\lambda)=1-\tilde{\lambda}dt+o(dt)
\]
For $x_{t+dt}^i\neq x_t^i$ and $x_t^1\neq x_t^2$,
\begin{align*}
\Prob(\xi_{t+dt}^i\in\{x_{t+dt}^i\}\times A,\xi_{t+dt}^{3-i}=\xi_t^{3-i},\Lambda_{t+dt}=&\Lambda_t|\xi_t^i=(x_t^i,u_t^i),\Lambda_t=\lambda)=\\
&=\lambda_i P^{S}_{x_{t+dt}^i-x_t^i}(u_t^i,A)dt+o(dt)
\end{align*}
However, when $\eta_t^1=\eta_t^2=x$,
\begin{align*}
\Prob(\xi^i_{t+dt}=\{x+z^i\}\times A^i,\Lambda_{t+dt}\in A_3&|\xi^i_t=(x,u_t^i),\Lambda_t=\lambda)=\\
&=\tilde{\lambda}\tilde{C}_{z^1,z^2}(\lambda,A_3,A^1,A^2)dt+o(dt)
\end{align*}
Everything else is just $o(dt)$. We launch the process from the initial state $J_0=(0,u_0^1,0,u_0^2,\lambda_0)$.

Denote the time of the first jump after $t$ with
\[
t_{\rm fj}(t)=\inf\{s>t: J_s\neq J_t\}
\]

Our main result will concern the four-tuple
\[
\tilde{J}_t=(\xi_t^1,\xi_t^2)=(\eta_t^1,\epsilon_t^1,\eta_t^2,\epsilon_t^2)
\]
on the set $(\mathbb{Z}^2\times S)^2$. We will prove weak convergence on the space $(\mathbb{R}^2\times S)^2$ endowed with the metric
\[
d((x^1,u^1,x^2,u^2),(y^1,w^1,y^2,w^2))=\sum_{i=1}^2(|x^i-y^i|+d_S(u^i,w^i))
\]
where $d_S(u^i,w^i)$ is the length of the shorter arc joining $u^i$ and $w^i$ on $S$.

For convenience, we define the arithmetic operations on $(\mathbb{R}^2\times S)^2$ by
\[
(x^1,u^1,x^2,u^2)+(y^1,w^1,y^2,w^2)=(x^1+y^1,u^1,x^2+y^2,u^2)
\]
\[
c(x^1,u^1,x^2,u^2)=(cx^1,u^1,cx^2,u^2)\qquad c\in\mathbb{R}
\]

\subsection{Main result}

\subsubsection{Notations and conditions}\hfill\label{Notations_and_Conditions}

We start with arbitrary dimension $d\ge 1$. For a single RWwIS, introduce the operator valued expected displacement and further the analog for the covariance

\begin{equation}\label{eq:moments}
M_l=\sum_{x\in\mathbb{Z}^d}x_lP_x\qquad\Sigma_{l,m}=\sum_{x\in\mathbb{Z}^d}x_lx_mP_x\qquad 1 \le l, m \le d
\end{equation}
where $x_l=(x,e_l)$ and from now on $P_x\equiv P_x^S$. It is easy to see that e.g.
\[
(M_l\mathbb{1})(u)=\mathbf{E}\left((\eta_{t_{\rm fj}(t)})_l-(\eta_t)_l|\varepsilon_t=u\right)
\]
where $\mathbb{1}\in L_{\infty}([a,b])$ is the constant 1 function on $[a,b]$ and $t_{\rm fj}(t)$ is the one walk variant of the above $t_{\rm fj}$ (it will always be clear from the context which is to use). Higher conditional moments can be defined analogously. Due to the bounded range condition below, all these moments are finite.

\noindent{\bf Conditions on the RWwIS}
{\it \begin{enumerate}

\item (Spectral gap) For the operator $Q_S: L_\infty(S)\to L_\infty(S)$ $1$ is a single eigenvalue and the remaining part of its spectrum lies in a circle $|z| \le \delta < 1$. Consequently, for the operator $Q_S^*$ acting over $\mathcal M(S)$, the space of finite measures over $S$, $1$ is also a single eigenvalue whose eigenfunction is the unique stationary probability measure $\rho$, i. e. $Q_S^*\rho=\rho$ (this can also be written $\rho Q_S= \rho$).
\item  (No drift) For $\forall 1 \le l \le d$
\[
(\rho,M_l\mathbb{1})=0;
\]
\item (Bounded range) $P_x = 0$ if $|x| > 1$ (for simplicity);

\item (Nonsingularity of the asymptotic covariance matrix) Along the lines of the LLT of \cite{KSz83}, for the diffusive limit of the displacements of the RWwIS, the asymptotic covariances of the coordinate vectors are equal to
\begin{eqnarray*}
\sigma_{lm}=(\rho,\Sigma_{lm}\mathbb{1})-(\rho,M_l(Q_S-I)^{-1}M_m\mathbb{1})-\\
-(\rho,M_m(Q_S-I)^{-1}\mathbb{1})
\end{eqnarray*}
$(1 \le l, m \le d)$. It is assumed that the matrix $\left(\sigma_{lm}\right)_{1 \le l, m \le d}$ is positive definite.
\end{enumerate}}

\noindent{\bf Conditions on the collision kernel}
{\it \begin{enumerate}
\item (Ergodicity) For $A\subseteq\mathcal{\tilde{I}}$ set $g(\lambda_-,A)=\sum_{z^1,z^2}\tilde{C}_{z^1,z^2}(\lambda,A,S,S)$. We assume that the homogeneous Markov chain defined by this kernel is an ergodic Harris chain (cf. \cite{R84}) with stationary distribution distribution $\rho_s$.
\item (Bounded range) Also $\tilde{C}_{z^1,z^2}=0$ for $|z^1|,|z^2|>1$
\end{enumerate}}

\subsubsection{Main result}

The main result of this paper concerns the limit distribution of the two interacting random walkers described in subsection 2.1.

\begin{theorem}\label{main_thm_1}

For every initial distribution of $(\xi^1_0, \xi^2_0)$, the density function of the weak limit law of $\tilde{J}/\sqrt{t}$ in $((\mathbb{R}^2\times S)^2,d)$ exists and is equal to
\[
h(x_1,u_1,x_2,u_2)=\frac{\rho(u_1)\rho(u_2)}{(2\pi)^2|\sigma|}\int_0^{1}\frac{1}{\lambda\sqrt{1-\lambda^2}}e^{-\frac{1}{2}\left(\frac{x_1^T\sigma^{-1}x_1}{\lambda}+\frac{x_2^T\sigma^{-1}x_2}{\sqrt{1-\lambda^2}}\right)}d\rho_s(\lambda)
\]
where $\rho$ is the stationary density of the internal states on $S$.
\end{theorem}

\subsubsection{Higher dimensions}
In higher dimensions, the treatment is essentially the same as well as the result with one exception. If $F$ is the probability that two particles starting from the same place will meet again, then instead of $\rho_s$ we have to use the mixture
\[
\tilde{\rho}_{\lambda_0}(A)=(1-F)\sum_{n=0}^{\infty}g^n(\lambda_0,A)F^n
\]
where $g^n$ is the $n$ step kernel (Note that $F<1$ if $d\ge 3$).

Note that $F$ depends on $\varepsilon_0^1$ and $\varepsilon_0^2$ but for the same reasons as above, we can assume that it is in principle possible to average over them with respect to some certain distribution. This problem is strongly related to the question when one asks what the distribution of the internal states is at the first return to the origin in the case of a single random walker. We do not discuss this in further detail.

\subsection{Methods}

\subsubsection{Local limit theorem for RWwIS's}\label{loc_lim_disc_time}

Our first step will be to generalize the local limit theorem described in \cite{KSz83} to continuous time and continuous internal states. Since now we only investigate the collision-free motion of one particle, the velocity magnitude $\lambda$ will be constant, thus we will consider it as a parameter and the internal state space will be $S$.

\begin{theorem}[Local Limit Theorem]\label{thm:Local_Limit_Theorem}
With the assumptions in \ref{Notations_and_Conditions}, namely
\begin{enumerate}
\item[(i)] $Q_S$ is ergodic and aperiodic with stationary distribution $\rho$;
\item[(ii)] $(\rho,M_l\mathbb{1})=0$ for every $1\leq l\leq d$;

\item[(iii)] The matrix $\sigma=(\sigma_{lm})_{1\leq l,m \leq d}$ whose elements are
\begin{align*}
\sigma_{lm}=(\rho,\sigma_{lm}\mathbb{1})-(\rho,M_l(Q_S-I)^{-1}M_m\mathbb{1})-\\
-(\rho,M_m(Q_S-I)^{-1}\mathbb{1})
\end{align*}
is positive definite.
\end{enumerate}
we have for every $A\subseteq S$
\begin{equation*}
\sum_{x\in\integers^d}\left|h_{t,x}(A|\xi_0=(0,u_0))-\frac{\rho(A)}{(\lambda t)^{d/2}}g_{\sigma}\left(\frac{x}{\sqrt{\lambda t}}\right)\right|=\mathcal{O}\left((\lambda t)^{-(d+1)/2}\right)
\end{equation*}
where
\[
h_{t,x}(A)=\Prob(\xi_t\in\{x\}\times A)
\]
and $g_{\sigma}$ denotes the density function of the $d$ dimensional normal distribution with mean $0$ and covariance matrix $\sigma$. The remainder term is uniform in $u_0$.
\end{theorem}

\begin{remark}
For the sake of generality we will prove this theorem without the bounded range condition only assuming that all moments like \eqref{eq:moments} are finite and the minimal lattice $\mathcal L$ determined by the jumps of the RWwIS coincides with  $\mathbb Z^d$ (cf. local CLT of \cite{KSz83}). (This condition is called the triviality of arithmetics.)
\end{remark}

We will show that -- for our case $d=2$ -- this implies
\begin{corollary}\label{tail}
Assume $d=2$. If $\tau=\inf\{t>t_{\rm fj}(0):\eta_t=0\}$, then $\forall u_0 \in S$
\begin{equation}\label{eq:First_Return1}
F_{\lambda,u_0}(t)=\Prob(\tau<t|\xi_0=(0,u_0))=1 - \frac{2\pi\sqrt{|\sigma|}}{\log (\lambda t)}+\mathcal{O}\left(\frac{\log\log(\lambda t)}{\log^2(\lambda t)}\right)
\end{equation}
where the remainder term is uniform in $u_0$.
\end{corollary}
Recall the definition of slowly varying functions
\begin{definition}
A positive function $L(t)$ defined on $\reals_+$ is slowly varying at infinity if
\[
\frac{L(ct)}{L(t)}\to 1\qquad\forall c\in\reals_+
\]
\end{definition}
Clearly, the tail of the distribution function $F_{\lambda,u_0}(t)$ is a slowly varying function, which proves to be a crucial property later.

By Definition \ref{speed_param},  $\lambda$  is just  the speed of the random walker and  the property
\begin{equation}\label{eq:scaling}
F_{\lambda,u_0}(t)=F_{1,u_0}(\lambda t)
\end{equation}
is evident by rescaling.

\begin{remark}
As it will be shown in Section \ref{STRP}, these excursions are very long, so it is plausible to assume by Theorem \ref{thm:Local_Limit_Theorem} that instead of \eqref{eq:First_Return1}, we only have to deal with the family of functions
\[
F_{\lambda}(t)=\int F_{\lambda,{u_0}}(t)d\rho(u_0)
\]
The existence of this average is granted by the uniformity of the above expansion of $F_{\lambda,u_0}$ in $u_0$. Rigorously, for every $u_0$, $(1-F_{\lambda}(t))/(1-F_{\lambda, u_0}(t))\to 1$ uniformly as $t\to\infty$, and this is sufficient in the sequel (cf. the proofs in the appendix).
\end{remark}

\begin{corollary}[Central Limit Theorem]\label{CLT}
With the assumptions of the local theorem,
\[
\left(\frac{\eta_t}{\sqrt{t}},\epsilon_t\right)\Rightarrow \mathcal{N}^d(0,\lambda\sigma)\times\rho\qquad\textrm{in}\quad (\mathbb{R}^d\times S,d_0)
\]
where $\mathcal{N}^d(0,\lambda\sigma)$ is the $d$-dimensional normal distribution with mean $0$ and covariance matrix $\lambda\sigma$. The metric $d_0$ is the sum of the euclidean metric on $\mathbb{R}^d$ and the previously defined $d_S$.
\end{corollary}

\subsubsection{Scaled Type Renewal Processes}\label{STRP}

Since the speed of each individual particle is conserved between consecutive collisions and the asymptotic inter-collision times satisfy \eqref{eq:scaling}, the idea naturally arises that one should deal with a renewal process, where the renewal times come from a one-parameter family of distributions like \eqref{eq:scaling}.

Consider a family of distribution functions $F_{\lambda}:$ $\lambda\in[a,b]\subset\reals_+$ with positive support and assume that $\forall \lambda\in[a,b]$ $F_{\lambda}(t) = F(\lambda t)$ for a nondegenerate distribution function $F$. The corresponding random variables are denoted by $X_{\lambda}$.

\begin{definition}

Suppose $(\lambda_0,\lambda_1, \lambda_2, \dots) \in [a, b]^\mathbb N$. Then the sequence $S_n = \sum_{j=0}^n X_{\lambda_j}:\ n = 0, 1, 2, \dots$ is called a {\rm scaled-type renewal process (STRP)} if $X_{\lambda_0},X_{\lambda_1}, X_{\lambda_2}, \dots$ is an independent sequence of random variables such that $\forall j \in \mathbb N$ the distribution of $X_{\lambda_j}$ is $F_{\lambda_j}$.

\end{definition}

\noindent\textbf{STRP with slow tail return times}

As it was previously mentioned, in the model under investigation the return times are very long, more specifically they satisfy the slowly varying tail property.

When  $\forall i$\ $\lambda_i=1$ (or any constant), there are many well known results, among which the one revealing the core of the phenomena is the following (cf. \cite{HM91}).

\begin{theorem}\label{Haeusler-Mason}
Let $X_i\geq 0$ be random variables with common distribution function F. For every $k\geq 0$, the following three statements are equivalent.

\begin{itemize}
\item $1-F(x)=L(x)$, where $L$ is a slowly varying function
\item $X_{n-k-1,n}/X_{n-k,n}\stackrel{P}{\to}0$ as $n\to\infty$
\item $\sum_{i=1}^{n-k}X_{i,n}/X_{n-k,n}\stackrel{P}{\to}1$ as $n\to\infty$
\end{itemize}
where $X_{i,n}$ is the ordered statistics from $X_i$ $i=1,..,n$.
\end{theorem}
\noindent In other words, the largest return times dominate the whole process. In fact, we do not need this theorem neither its generalization (however intuitively it should hold), the important result is related to the age and residual age process.

Now consider the STRP with an arbitrary sequence of parameters $\Lambda=(\lambda_0,\lambda_1,...)$ from $[a,b]$. Let $N_{t}=\max\{n:S_{n}\leq t\}$. The random variables
\[
Y_{t}=t-S_{N_t}\qquad Z_{t}=S_{N_t+1}-t
\]
are called {\it the age and the residual age} respectively. For these quantities, we have

\begin{theorem}\label{thm:slowly_var}

If $F$ has a slowly varying upper tail, then as $t\to\infty$,
\[
\frac{Y_{t}}{t}\stackrel{P}{\to}1\qquad \frac{Z_{t}}{t}\stackrel{P}{\to}1\qquad
\]
\end{theorem}

This theorem was proved in a more general setting in \cite{PGySz2010} but in our case a much simpler proof is available which we present in the appendix.

\vspace{3mm}
\noindent\textbf{Markovian Renewal Processes}

\begin{definition}
A scaled-type renewal process $S_{\lambda_0,n} = \sum_{j=0}^n X_{\lambda_j}:\ n = 0, 1, 2, \dots$ is called a {\rm scaled type Markovian renewal process (STMRP)} if $\lambda_0, \Lambda_1, \Lambda_2, \dots$ is a homogeneous Markov chain with values in $[a, b]$ and for every realization  $\lambda_1, \lambda_2, \dots$ of this Markov chain $S_n = \sum_{j=0}^n X_{\lambda_j}:\ n = 0, 1, 2, \dots$ is a scaled-type renewal process.
\end{definition}

The essential properties of such processes are well described in the literature when the waiting times have finite means (cf. the introduction of \cite{PGySz2010} for further reference). However, it is clear from \ref{tail} that we are now facing the infinite mean case which seemed untouched before the authors established results in \cite{PGySz2010} with further restrictions on $F$. This was done in a more general setting, here we only present what is necessary for our current purposes. In our case when $F$ is slowly varying, the more complicated machinery of the cited paper is not necessary and just as in the case of Theorem \ref{thm:slowly_var}, we give a much simpler proof in the appendix.

Let $g(\lambda_-,\lambda_+)$ be the transition kernel of the Markov-chain. Suppose that this is a recurrent Harris chain with stationary measure $\rho_s$.

The expectation of $X_{\lambda}$ is denoted by $\mu_{\lambda}=\mu/\lambda$ whenever $\mu=\int_0^{\infty} x dF$ is finite. We repeat that the parameter interval is chosen so that $0<a<b<\infty$.

Let $N_{t,\lambda_0}$ denote the number of the renewals occurred before time $t$ (including the one at $t=0$) with initial parameter value $\lambda_0$, i.e.
\begin{equation}\label{felujitasok_szama}
N_{t,\lambda_0}=\inf\{n:S_{\lambda_0,n}\geq t\}
\end{equation}
and let $U_{\lambda_0}(t)$ be its expectation.

A classical question is: what is the "type" of the current renewal at time $t$, i.e. what is the distribution of the parameter $\lambda$. Denote the corresponding measure conditioned on the initial parameter value $\lambda_0$, by $\Phi_{t,\lambda_0}$, i.e.
\[
\Phi_{t,\lambda_0}(A)=\Prob(\Lambda_{N_{t,\lambda_0}-1}\in A\subseteq[a,b])
\]

By investigating the asymptotics we get
\begin{theorem}\label{thm:inf_exp}\label{limstatdistr}
If  $1-F$ is slowly varying, then
\[
\lim_{t\to\infty}\Phi_{t,\lambda_0}(A)=\rho_s(A)
\]
\end{theorem}
Intuitively, this means that the waiting times are so similar in a probabilistic sense, that the rescaling does not matter asymptotically. This means that the process behaves analogously as if the waiting times were iid.

\subsubsection{Sketch of proof of the main result}
For the convenience of the reader, here we present the main ideas used in the proof of Theorem \ref{main_thm_1}.

\begin{enumerate}

\item First note, that the spatial difference of the two random walkers $\xi'_t=(\eta_{t,\lambda_0}^1-\eta_{t,\lambda_0}^2,\epsilon_t^1,\epsilon_t^2,\Lambda_t)$ is again a RWwIS (modulo the origin, where the collision kernel also comes into play) on the state space
\[
\integers^2\times S^2\times \tilde{I}
\]
where $\tilde{I}$ is as in Section 2.1.2. The collision of the two particles corresponds to the return of this walk to the origin and we can use our results developed for the return times since the behavior of the first return is not effected by the dynamics at the origin.
\item At time $t$, the relevant information for our goal is the state of the process at the last collision before $t$ and its history since then.
\item It can be shown that, by dividing by $\sqrt{t}$, the location of the last collision before time $t$ goes to zero in probability. Due to Theorem \ref{thm:slowly_var}, the amount of time elapsed since the last collision dominates the whole process. Consequently, we only have to treat two RWwIS which evolve conditioned on not meeting.
\item Now it is clear that the limit distribution will be a mixture according to the value of the outgoing $\Lambda$ at the last collision before $t$. Due to recurrence, the particles will meet infinitely many times and therefore the asymptotic distribution will be $\rho_s$ by Theorem \ref{limstatdistr}.

\item Finally, we have to derive the limit distribution of two independent RWwIS conditioned on not meeting. The unconditional limit is the product of the independent limits determined by Corollary \ref{CLT} and we will show that the conditioned one is the same since the condition becomes irrelevant asymptotically. (This is the point where our method brakes down in one dimension). This can be shown by defining an appropriate random time which behaves like stopping time (although it is not), so the strong Markov-property can be used (Lemma \ref{Bolthausen_identity}). Since this quasi-stopping time is very small in a certain sense, we can obtain the desired result. However, we will choose a different approach and refer to the correspondig result for ordinary random walks.
\end{enumerate}

\section{Proof of Local Limit Theorem and of related results}\label{Proof_LCLTH}

We will follow the main ideas outlined in \cite{KSz83} but for the sake of self containedness, we conduct the whole proof except for some tedious calculation.

The transition operator for the discrete time RWwIS is
\[
T:L_{\infty}(\mathbb{Z}^d\times S)\to L_{\infty}(\mathbb{Z}^d\times S)
\]
with
\[
(Tf)(x)=\sum_{y\in\mathbb{Z}^d\backslash\{0\}}P_yf(x-y)
\]
for $f\in L_{\infty}(\mathbb{Z}^d\times S)$. Clearly if $T_i$ is the time when the $i$th transition occurs i.e.
\[
T_0=0\qquad T_i=\inf\{t>T_{i-1}:\eta_t\neq\eta_{T_{i-1}}\}
\]
then
\[
h_{T_{i+1},x}(.)=\sum_{y\in\mathbb{Z}^d\backslash\{0\}}P_y^*h_{T_i,x-y}(.)
\]
Note that the dual is $\mathcal{M}(\mathbb{Z}^d\times S)$  i.e. the signed measures of bounded total variation on $\mathbb{Z}^d\times S$.

As usually in case of limit theorems, we use spatial Fourier transforms
\[
\hat{f}:[-\pi,\pi)^d\to L_{\infty}(S)\qquad \hat{f}(s)=\sum_{x\in\mathbb{Z}^d}e^{i(s,x)}f(x)
\]
For the operator $T$, we have
\[
\widehat{(Tf)}(s)=\sum_{x\in\mathbb{Z}^d}e^{i(s,x)}(Tf)(x)=\sum_{x\in\mathbb{Z}^d}\sum_{y\in\mathbb{Z}^d\backslash\{0\}}e^{i(s,x)}P_yf(x-y)
\]
which by the change of variables $y'=y$ and $x'=x-y$ further equals
\begin{align*}
\sum_{y'\in\mathbb{Z}^d\backslash\{0\}}\sum_{x'\in\mathbb{Z}^d}e^{i(s,x'+y')}P_{y'}f(x')=\sum_{y'\in\mathbb{Z}^d\backslash\{0\}}e^{i(s,y')}P_{y'}\sum_{x'\in\mathbb{Z}^d}e^{i(s,x')}f(x')
\end{align*}
that is just $\alpha(s)\hat{f}(s)$ where we introduced the operator valued Fourier-transform
\[
\alpha(s)=\sum_{y'\in\mathbb{Z}^d\backslash\{0\}}e^{i(s,y')}P_{y'}
\]
By induction, it follows that
\begin{equation}\label{eq:Four_induct}
\widehat{(T^nf)}(s)=\alpha^n(s)\hat{f}(s)
\end{equation}

Returning to the continuous time, note that by Definition \ref{speed_param},
\[
\{T_i-T_{i-1}\}_{i=1}^{\infty}
\]
is an i.i.d. sequence with distribution $EXP(\lambda)$, thus the number of transitions occurred up until time $t$ has distribution $POI(\lambda t)$. Using this, we have for $A\subseteq S$
\begin{equation}\label{h_poi}
h_{t,x}(A|\xi_0=(0,u_0))=e^{-\lambda t}\sum_{n=0}^{\infty}\frac{(\lambda t)^n}{n!}((T^*)^n(\delta_{0}\delta_{u_0}))\bigg|_{x,A}
\end{equation}
Using the Fourier-inversion formula, this equals
\[
\frac{1}{(2\pi)^d}\int_{-\pi}^{\pi}...\int_{-\pi}^{\pi}e^{-i(s,x)}\sum_{n=0}^{\infty}e^{-\lambda t}\frac{(\lambda t)^n}{n!}\widehat{((T^*)^n(\delta_{0}\delta_{u_0}))}(s)ds\bigg|_{A}
\]
which after using \eqref{eq:Four_induct} becomes
\begin{align*}
\frac{1}{(2\pi)^d}\int_{-\pi}^{\pi}...\int_{-\pi}^{\pi}e^{-i(s,x)}&\sum_{n=0}^{\infty}e^{-\lambda t}\frac{(\lambda t)^n}{n!}(\alpha^*)^n(s)\delta_{u_0}ds\bigg|_{A}=\\
&=\frac{1}{(2\pi)^d}\int_{-\pi}^{\pi}...\int_{-\pi}^{\pi}e^{-i(s,x)}e^{\lambda t(\alpha^*(s)-1)}\delta_{u_0}ds\bigg|_A
\end{align*}
since $\widehat{\delta_0\delta_{u_0}}=\delta_{u_0}$.

\begin{proof}[Proof of Theorem \ref{thm:Local_Limit_Theorem} for $d=1$]

In addition to \eqref{eq:moments} we introduce the third moment
\[
\Xi=\sum_{x\in\mathbb{Z}^d\backslash\{0\}}x^3P_x
\]
following \cite{N09}. Using these moments
\begin{equation}\label{eq:alpha_expansion}
\alpha(s)=Q_S+isM-\frac{s^2}{2}\Sigma-\frac{is^3}{6}\Xi+o(s^3)
\end{equation}
Using a straightforward generalization of Theorem 2.9 in Chapter VIII in \cite{K66} we have that its largest eigenvalue $\chi(s)$ has a similar expansion
\[
\chi(s)=1+r_1s+\frac{r_2}{2}s^2+\frac{r_3}{6}s^3+o(s^3)
\]
The coefficients are
\[
r_1=0\qquad r_2=-(\rho,\Sigma\mathbb{1})+2(\rho,M(Q_S-I)^{-1}M\mathbb{1})=-\sigma^2
\]
as calculated in \cite{KSz83} and
\[
r_3=i\left(3(\rho,\Sigma(Q_S-I)^{-1}M\mathbb{1})+3(\rho,M(Q_S-I)^{-1}\Sigma\mathbb{1})-(\rho,\Xi\mathbb{1})\right)
\]
was computed in \cite{N09}. The largest eigenvalue of $\kappa(s)=e^{\alpha^*(s)-1}$ is
\[
e^{\overline{\chi(s)}-1}=1+r_1s+\left(r_1^2+\frac{r_2}{2}\right)s^2+\left(r_1^3+\frac{r_1r_2}{2}+\frac{\overline{r}_3}{6}\right)s^3+o(s^3)
\]
by elementary calculation which after plugging the above expressions becomes
\begin{equation}\label{eq:largest_eigenvalue}
1-\frac{\sigma^2s^2}{2}+\frac{\overline{r}_3}{6}s^3+o(s^3)
\end{equation}
We mention that if there would be drift, one should use the above formula with $r_1=i(\rho,M\mathbb{1})$ (cf. \cite{KSz83}).

Let $\varphi(s)$ be the eigenvector corresponding to the eigenvalue $\chi(s)$ with $(\rho,\varphi(s))=1$. Introduce the operator
\[
P_{\varphi}f=(\rho,f)\varphi(s)\qquad f\in L_{\infty}(S)
\]
It is easy to see that for $g\in\mathcal{M}(S)$
\[
P_{\varphi(s)}^*g=(g,\varphi(s))\rho
\]

Since every moment is finite, the perturbation is analytic and  $e^{\alpha(s)-1}$ is  continuous in $s$. Consider the operator valued function
\[
R(s)=e^{\alpha(s)-1}-e^{\chi(s)-1}P_{\varphi(s)}
\]
which is again continuous in $s$. Since by assumption there is a spectral gap for $Q_S$,
\[
||R(0)||=||e^{Q_S-1}-P_{\mathbb{1}}||<1
\]
Continuity implies $||R(s)||<1$ for sufficiently small $s$.

Since the eigenspaces depend continuously on the perturbation at $s=0$,
\[
e^{\alpha(s)-1}=\left(1-\frac{\sigma^2s^2}{2}+\frac{r_3}{6}s^3+o(s^3)\right)(P_{\mathbb{1}}+o(1))+R(s)
\]
and thus
$e^{\lambda t\left(\alpha\left(\frac{s}{\sqrt{\lambda t}}\right)-1\right)}$ can be decomposed
\begin{align}\label{al}
\nonumber\left(1-\frac{\sigma^2s^2}{2\lambda t}+\frac{r_3}{6}\frac{s^3}{(\lambda t)^{3/2}}+o\left(\frac{s^3}{(\lambda t)^{3/2}}\right)\right)^{\lambda t}&(P_{\mathbb{1}}+o(1))+\\
&+R^{\lambda t}\left(\frac{s}{\sqrt{\lambda t}}\right)
\end{align}
where the  last term is exponentially converging to zero.
Now we will show that $2\pi\sqrt{\lambda t}$ times
\begin{align}\label{local_limit_proof}
\nonumber\bigg|\bigg|\frac{1}{2\pi}\int_{-\pi}^{\pi}e^{-isx}&e^{\lambda t(\alpha^*(s)-1)}\delta_{u_0}ds\bigg|_A-\\
&-\frac{\rho(A)}{\sqrt{2\pi \lambda t}\sigma}e^{-\frac{x^2}{2\lambda t\sigma^2}}\left(1-\frac{ir_3}{6}\frac{x(3\sigma^2\lambda t-x^2)}{\sigma^6(\lambda t)^2}\right)\bigg|\bigg|
\end{align}
goes to zero. By $\widehat{xf}(s)=-is\hat{f}'(s)$, after some elementary calculations, one can get
\begin{align*}
e^{-\frac{x^2}{2\lambda t\sigma^2}}&\left(1-\frac{ir_3}{6}\frac{x(3\sigma^2\lambda t-x^2)}{\sigma^6(\lambda t)^2}\right)=\\
&=\sqrt{\lambda t}\frac{\sigma}{\sqrt{2\pi}}\int_{-\infty}^{\infty}e^{-itx}e^{-\frac{\lambda t\sigma^2s^2}{2}}\left(1+\frac{r_3\lambda t}{6}s^3\right)ds
\end{align*}
plugging this back and making the change of variables $s'=\sqrt{\lambda t}s$ we have
\begin{align*}
\bigg|\bigg|\int_{-\pi\sqrt{\lambda t}}^{\pi\sqrt{\lambda t}}&e^{-ix\frac{s'}{\sqrt{\lambda t}}}e^{\lambda t\left(\alpha^*\left(\frac{s'}{\sqrt{\lambda t}}\right)-1\right)}\delta_{u_0}ds'\bigg|_A-\\
&-\rho(A)\int_{-\infty}^{\infty}e^{-ix\frac{s'}{\sqrt{\lambda t}}}e^{\frac{-\sigma^2s'^2}{2}}\left(1+\frac{r_3}{6}\frac{s'^3}{\sqrt{\lambda t}}\right)ds'\bigg|\bigg|
\end{align*}
By the triangle equality, this integral is bounded from above by the sum of the following four terms (we drop the prime)
\[
I_1=\int_{|s|<(\lambda t)^{\epsilon}}\bigg|\bigg|e^{\lambda t\left(\alpha^*\left(\frac{s}{\sqrt{\lambda t}}\right)-1\right)}\delta_{u_0}\bigg|_A-\rho(A)e^{-\frac{\sigma^2s^2}{2}}\left(1+\frac{r_3}{6}\frac{s^3}{\sqrt{\lambda t}}\right)\bigg|\bigg|ds
\]
\[
I_2=\rho(S)C\int_{|s|\geq(\lambda t)^{\epsilon}}e^{-\frac{\sigma^2s^2}{2}}ds
\]
\[
I_3=\int_{(\lambda t)^{\epsilon}<|s|<\gamma\sqrt{\lambda t}}\bigg|\bigg|e^{\lambda t\left(\alpha^*\left(\frac{s}{\sqrt{\lambda t}}\right)-1\right)}\delta_{u_0}\bigg|\bigg|ds
\]
\[
I_4=\int_{\gamma\sqrt{\lambda t}<|s|<\pi\sqrt{\lambda t}}\bigg|\bigg|e^{\lambda t\left(\alpha^*\left(\frac{s}{\sqrt{\lambda t}}\right)-1\right)}\delta_{u_0}\bigg|\bigg|ds
\]
for $0<\epsilon<1/6$, $C=\sqrt{\left(1+\frac{r_3^2}{36}\right)}$ and $\lambda t>1$. Again by elementary calculations,
\begin{align*}
\left(1-\frac{\sigma^2s^2}{2\lambda t}+\frac{r_3}{6}\frac{s^3}{(\lambda t)^{3/2}}+o\left(\frac{s^3}{(\lambda t)^{3/2}}\right)\right)^{\lambda t}=\\
=e^{-\frac{\sigma^2s^2}{2}}\left(1+\frac{r_3}{6}\frac{s^3}{\sqrt{\lambda t}}+o\left(\frac{s^3}{\sqrt{\lambda t}}\right)\right)
\end{align*}

From \eqref{al} now it is clear that $I_1=o(1/\sqrt{\lambda t})$. Using the usual upper bound for the tail of the Gaussian function, it is easy to see that $I_2=o(1/\sqrt{\lambda t})$ as well while the trivial arithmetic condition - similarly as in \cite{KSz83} - implies the exponential decay of $I_4$ with growing $\lambda t$. If $\gamma$ is chosen small enough, then from \eqref{eq:largest_eigenvalue}, it follows that
\[
||e^{\alpha(s)-1}||<e^{-\frac{\sigma^2s^2}{4}}\qquad |s|<\gamma
\]
With this
\begin{align*}
I_3=&\sqrt{\lambda t}\int_{(\lambda t)^{\epsilon-1/2}<|s|<\gamma}\left|\left|e^{\lambda t(\alpha^*(s)-1)}\delta_{u_0}\right|\right|ds<\\
&<\sqrt{n}\int_{(\lambda t)^{\epsilon-1/2}<|s|<\gamma}e^{-\frac{\sigma^2s^2\lambda t}{4}}ds=\int_{(\lambda t)^{\epsilon}<|s|<\gamma\sqrt{\lambda t}}e^{-\frac{\sigma^2s^2}{4}}ds=o\left(\frac{1}{\sqrt{\lambda t}}\right)
\end{align*}
as before, so $I_1+I_2+I_3+I_4=o\left(\frac{1}{\sqrt{\lambda t}}\right)$ and the proof is ready.
\end{proof}
\begin{proof}[Proof for $d\geq 2$]

The multidimensional case is a straightforward generalization. For expansion of the largest eigenvalue
\[
e^{\chi(s)-1}=1-\frac{(s,\sigma s)}{2}+\sum_{i=1}^d\sum_{j=1}^d\sum_{k=1}^d\frac{r_{3,i,j,k}}{6}s_is_js_k+o(|s|^3)
\]
while one has to prove the convergence of
\begin{align*}
\frac{(\lambda t)^{d/2}}{(2\pi)^d}\bigg|\bigg|\int_{-\pi}^{\pi}...\int_{-\pi}^{\pi}e^{-i(x,s)}e^{\lambda t(\alpha^*(s)-1)}\delta_{u_0}ds\bigg|_A-\\
-\frac{\rho(A)}{(\lambda t)^{d/2}}\int_{-\infty}^{\infty}...\int_{-\infty}^{\infty}e^{-\frac{\lambda t(s,\sigma s)}{2}}e^{-i(x,s)}(1+\lambda tf(s))ds\bigg|\bigg|
\end{align*}
to zero where $f(s)$ is the above term containing the $r_{3,i,j,k}$-s. It turns out that the term containing $f(s)$ is just $O\left(1/\sqrt{\lambda t}\right)$ and Theorem \ref{thm:Local_Limit_Theorem} follows.
\end{proof}
\begin{proof}[Proof of Corollary \ref{tail}]

We will obtain the formula \eqref{eq:First_Return1} through the discrete time process. The discrete process was defined in Definition \ref{def:Sinai}. As a straightforward generalization of Theorem 2 of \cite{N09} using the elements of the previous proof, we have
\begin{align*}
h_{n,x}(A)&=\Prob(\xi_n\in\{x\}\times A|\xi_0=(0,u_0))=\\
&=\frac{1}{n^{d/2}}\rho(A)g_{\sigma}\left(\frac{x}{\sqrt{n}}\right)+\mathcal{O}\left(n^{-\frac{d+1}{2}}\right)
\end{align*}
For $d=2$, the remainder term is $\mathcal{O}(n^{-3/2})$ which means it is summable. From the proof of  Theorem 6 of \cite{N09}, we know that
\[
f_n=\Prob(\tau>n|\eta_0=0)=\frac{2\pi\sqrt{|\sigma|}}{\log n}+\mathcal{O}\left(\frac{\log\log n}{\log^2 n}\right)
\]
with the remainder term being uniform in the initial state.
From now $C:=2\pi\sqrt{|\sigma|}$. By the same argument that led to \eqref{h_poi},
\[
1-F_{\lambda}(t)=e^{-\lambda t}\sum_{n=0}^{\infty}\frac{(\lambda t)^n}{n!}f_n
\]

First, introduce a cutoff

\[
\bigg|1-F(t)-\sum_{n=\lceil\frac{\lambda t}{2}\rceil}^{\lfloor\frac{3}{2}\lambda t\rfloor}f_ne^{-\lambda t}\frac{(\lambda t)^n}{n!}\bigg|\leq\mathbf{P}\bigg(|N_t-\lambda t|\geq\frac{\lambda t}{2}\bigg)<\frac{4}{\lambda t}
\]
using Chebyshev's inequality. For the remaining $n$'s it can be easily obtained that
\[
\frac{\log(\lambda t)}{\log n}=1+\mathcal{O}\left(\frac{1}{\log(\lambda t)}\right)
\]
Using this and the Chebyshev inequality again after some calculation we obtain
\begin{align*}
\sum_{n=\lceil\frac{\lambda t}{2}\rceil}^{\lfloor\frac{3}{2}\lambda t\rfloor}f_ne^{-\lambda t}\frac{(\lambda t)^n}{n!}=&\frac{C}{\log(\lambda t)}+\mathcal{O}\left(\frac{1}{\lambda t\log(\lambda t)}\right)+\\
&+\mathcal{O}\left(\frac{1}{\log^2(\lambda t)}\right)+\mathcal{O}\left(\frac{1}{\lambda t\log^2(\lambda t)}\right)+RT
\end{align*}

The remainder term can be estimated similarly. With some elementary calculation again,
\[
\frac{\log\log n}{\log\log(\lambda t)}=1+\mathcal{O}\left(\frac{1}{\log(\lambda t)\log\log(\lambda t)}\right)
\]
and by short computation we obtain
\[
RT=\mathcal{O}\left(\frac{\log\log(\lambda t)}{\log^2(\lambda t)}\right)
\]

Finally one can observe that this last one is the slowest error term so
\[
1-F(t)=\frac{2\pi\sqrt{|\sigma|}}{\log(\lambda t)}+\mathcal{O}\left(\frac{\log\log(\lambda t)}{\log^2(\lambda t)}\right)
\]
as desired.
\end{proof}
\begin{proof}[Proof of Corollary \ref{CLT}]
The result could be derived directly from the local theorem but in one dimension, it is simpler to let $t\to\infty$ in \eqref{al} to obtain
\[
e^{\lambda t\left(\alpha^*\left(\frac{s}{\sqrt{\lambda t}}\right)-1\right)}\delta_{u_0}\to e^{-\frac{\sigma^2 s^2}{2}}P^*_{\mathbb 1}\delta_u
\]
where by definition $P_{\mathbb{1}}^*\delta_u=\rho$.

The multidimensional theorem can be obtained similarly.
\end{proof}
\section{Proof of Theorem 1}

First we want to transfer our results for the returns of one random walker to the origin to the collisions of the two particle case. Note that the spatial difference of the two random walkers form a one particle RWwIS $\xi'_t=(\eta_{t,\lambda_0}^1-\eta_{t,\lambda_0}^2,\epsilon_t^1,\epsilon_t^2,\Lambda_t)$ on the state space
\[
(\mathbb{Z}\times S)^2\times \tilde{I}
\]
except for the origin where the collision kernel spoils the translation invariance. The component in $\tilde{I}$ is the speed-parameter $\Lambda_t$ of the first particle, while the other one's is determined by the fixed total energy ($2E=1$). However the speed-parameter of this difference-process is the sum of these two individual parameters (cf. superposition of Poisson processes)
\[
\Lambda_t+\sqrt{1-\Lambda_t^2}
\]
This is clearly bounded away from zero except for the trivial zero energy case, so we can use the results in Section \ref{STRP}.

Although the behavior is different in the origin, the tail of the times between consecutive visits to the origin of $\xi_t'$ (collisions) will still be slowly varying.

Thus if $\tau(t)$ is now the time of the last collision before $t$, i.e.
\[
\tau(t)=\sup\{s\leq t|\exists\epsilon>0:\eta_{u,\lambda_0}^1\neq \eta_{u,\lambda_0}^2, \eta_{s,\lambda_0}^1=\eta_{s,\lambda_0}^2 \quad s-\epsilon<u<s\}
\]
then $\tau(t)/t\stackrel{P}{\to} 0$ by Theorem \ref{thm:slowly_var}.

Now we have the following decomposition of the process
\begin{equation}\label{decomposition}
\frac{\tilde{J}_{t}}{\sqrt{t}}\bigg|_{J_0}=\sqrt{\frac{t-\tau(t)}{t}}\frac{\tilde{J}_{t-\tau(t)}}{\sqrt{t-\tau(t)}}\bigg|_{J_{\tau(t)}}+\frac{\tilde{J}_{\tau(t)}}{\sqrt{t}}\bigg|_{J_0}
\end{equation}
where we indicated the starting states in subscripts. The last term can be dealt with using Theorem 4.1 in \cite{B68} and the following
\begin{lemma}\label{place_of_last_coll}
For the common place of the last collision $\eta_{\tau(t)}^c\equiv\eta^i_{\tau(t)}$,
\[
\frac{\eta_{\tau(t)}^c}{\sqrt{t}}\stackrel{P}{\to} 0
\]
\end{lemma}

\begin{proof}
Pick $\epsilon,\delta>0$ and let $t$ be large enough such that $\Prob(\tau(t)>\epsilon t)<\delta$. Then
\[
\Prob\left(\frac{|\eta_{\tau(t)}^c|}{\sqrt{t}}>\epsilon\right)<\int_0^{\delta}\Prob(|\eta_{ut}^c|>\epsilon\sqrt{t}|\tau(t)=ut)\mu_{\tau(t)}(tdu)+\delta
\]
where $\mu_{\tau(t)}$ is the measure generated by $\tau(t)$.

Since
\[
\{|\eta_{ut}^c|>\epsilon\sqrt{t}\}\cap\{\tau(t)=ut\}\subseteq\{|\eta_{ut}^1|>\epsilon\sqrt{t}\}\cap\{\tau(t)=ut\}
\]
the integrand can be bounded from above by
\begin{align}\label{ineq1}
\Prob(|\eta_{ut}^1|&>\epsilon\sqrt{t}|\tau(t)=ut)\leq\\
\nonumber &\leq\sup_{\lambda,u^i}\Prob(|\eta_{ut}^1|>\epsilon\sqrt{t}|\eta_{ut}^1=\eta_{ut}^2,\Lambda_{ut}=\lambda,\epsilon_{ut}^i=u^i)
\end{align}
The inequality is due to the fact that
\[
\{\tau(t)=s\}=\{\eta_s^1=\eta_s^2\}\cap\{\eta_u^1\neq\eta_u^2 \quad u\in(s,t]\}
\]
where the second event only affects the distribution of $\eta_s^1$ through the internal states. Obviously, the condition $\eta_{ut}^1=\eta_{ut}^2$ is restrictive spatially so we can further estimate \eqref{ineq1} by
\begin{align*}
\sup_{\lambda,u^i}\Prob(|\eta_{ut}^1|>\epsilon\sqrt{t}|\Lambda_{ut}=\lambda,\epsilon_{ut}^i=u^i)\leq\\
\leq\frac{1}{\epsilon^2 t}\sup_{\lambda,u^i}\mathbb{D}^2(|\eta_{ut}^1||\Lambda_{ut}=\lambda,\epsilon_{ut}^i=u^i)
\end{align*}
using Chebysev's inequality. Note that by the diffusive nature of the process, the second moment is monotonously increasing with time, and it is the largest if $\eta_1$ posesses all the energy throughout the whole process up to time $ut$. Thus, the further bound can be obtained:
\begin{align*}
\frac{1}{\epsilon^2 t}\sup_{u^i}\mathbb{D}^2(|\eta_{\delta t}^1|&|\Lambda_{s}=2E\quad s\in[0,\delta t],\epsilon_{\delta t}^i=u^i)=\\
&=\frac{1}{\epsilon^2 t}\sup_{u^1}\mathbb{D}^2(|\eta_{\delta t}^1||\Lambda_{s}=2E\quad s\in[0,\delta t],\epsilon_{\delta t}^1=u^1)
\end{align*}
since the second particle is then standing still at the origin. This variance is nothing else but the variance of a single RWwIS  $\eta_{\delta t}$ with rate $2E$ and its internal state conditioned to be $u^1$ at $\delta t$.

By the bounded range condition, the above variance can be estimated from above by the variance of an ordinary continuous time random walk, in which the one step variance is $1$. It is well known that the variance of such a process at $\delta t$ is $2E\delta t\tilde{\sigma}$ for some constant $\tilde{\sigma}$.

Using this and that $\mu_{\tau(t)}([0,t\delta])<1$ we have that
\[
\Prob\left(\frac{|\eta_{\tau(t)}^c|}{\sqrt{t}}>\epsilon\right)<\delta(\epsilon^{-2}2E\tilde{\sigma}+1)
\]
for some constant $K$ independent of $\delta$. Since $\delta$ is arbitrary, the proof is finished.
\end{proof}

As the next step, note that $\sqrt{(t-\tau(t))/t}\stackrel{P}{\to} 1$ by Theorem \ref{thm:slowly_var} and by Theorem 4.4 in \cite{B68} we only have to prove the weak limit of the remaining term in \eqref{decomposition}.

This limit is nothing else but the joint limit of two continuous time RWwIS starting from the origin and not meeting once they depart. The distribution of the energy between these two are according to the distribution of $\Lambda_{\tau(t)}$ and by Theorem \ref{limstatdistr}, we have
\[
\Prob(\Lambda_{\tau(t)}\in A\subseteq\tilde{I})\to\rho_s(A)
\]

Thus we are ready if we can show that

\begin{lemma}\label{fin_lem}
In $d=2$, the joint law of two independent RWwIS (starting from the origin) conditioned on not meeting once they depart  is the product of the independent one particle limit-laws, i.e. for continuity sets (of the appropriate measures) $A_1$ and $A_2$,
\begin{align}
\Prob\left(\left(\frac{\eta_t^i}{\sqrt{t}},\epsilon_t^i\right)\in A_i\quad i=1,2\bigg|\eta_s^1\neq\eta_s^2\quad s\in [t_{\rm fj}(0),t]\right)\to\\
\to\prod_{i=1}^2\lim_{t\to\infty}\Prob\left(\left(\frac{\eta_t^i}{\sqrt{t}},\epsilon_t^i\right)\in A_i\right)
\end{align}
\end{lemma}

Before we proceed with the proof, we establish an identity first. Introduce
\[
T(t)=\inf\{s>0:\eta_s^1=\eta_s^2,\eta_u^1\neq\eta_u^2\quad t_{\rm fj}(s)\leq u\leq s+t\}
\]
Note that if there is a constant period of length $t$ at the same site before a ''depart and not return'', this definition of $T(t)$ gives its starting time.
Note that $\Prob(T(t)<\infty)=1$ and that $\{T(t)>a\}\subseteq\{\tau(a)<t\}$ and thus
\[
\Prob\left(\frac{T(t)}{t}>\epsilon\right)\leq\Prob\left(\frac{\tau(\epsilon t)}{\epsilon t}<\frac{1}{\epsilon}\right)\to 0
\]
In other words, $T(t)/t\stackrel{P}{\to} 0$.

Let $\tilde{J}_t^{id}$ denote the non-interacting two-particle system which makes it's first jump according to the collision kernel and note that then
\[
[\tilde{J}_t|\eta_s^1\neq\eta_s^2\quad s\in[t_{\rm fj}(0),t]]\stackrel{D}{=}[\tilde{J}_t^{id}|\eta_s^1\neq\eta_s^2\quad s\in[t_{\rm fj}(0),t]]
\]
We have
\begin{lemma}\label{Bolthausen_identity}
For any $A\in\mathcal{B}((\mathbb{Z}^2\times S)^2)$
\begin{align*}
\Prob(\tilde{J}_t^{id}\in A|&\eta_s^1\neq\eta_s^2\quad s\in[t_{\rm fj}(0),t])=\\
&=(1-e^{-\tilde{\lambda} t})\Prob(\tilde{J}_{T(t)+t}^{id}-\tilde{J}_{T(t)}^{id}\in A|\epsilon^i_{T(t)}=\epsilon_0^i)+\delta_{J_0}e^{-\tilde{\lambda} t}
\end{align*}
where again $\tilde{\lambda}=\lambda+\sqrt{1-\lambda^2}$.
\end{lemma}
\begin{proof}
Assume first that $t_{\rm fj}(0)<t$ and set
\[
L_s=\cup_{u\in[0,s-t]}\{\eta_u^1=\eta_u^2, \eta_r^1\neq\eta_r^2, t_{\rm fj}(u)\leq r\leq u+t\}
\]
and let $\mu_t$ denote the distribution of $T(t)$. Then
\begin{align*}
\Prob(\tilde{J}_{T(t)+t}^{id}-\tilde{J}_{T(t)}^{id}\in A&|\epsilon_{T(t)}^i=\epsilon_0^i)=\\
&=\int\Prob(\tilde{J}_{s+t}^{id}-\tilde{J}_s^{id}\in A|T(t)=s,\epsilon_s^i=\epsilon_0^i)d\mu_t(s)
\end{align*}
which equals
\[
\int\Prob(\tilde{J}_{s+t}^{id}-\tilde{J}_s^{id}\in A|L_s^c,\eta_s^1=\eta_s^2, \eta_r^1\neq\eta_r^2\quad r\in[t_{\rm fj}(s),s+t], \epsilon_s^i=\epsilon_0^i)d\mu_t(s)
\]
Since our process is Markov, $L_s^c$ is superfluous while the remaining integrand is time-translational invariant, so by $\Prob(T(t)<\infty)=1$, $\eta_0^1=\eta_0^2=0$ and the way we defined the substraction, we obtain
\[
\Prob(\tilde{J}_t^{id}\in A|\eta_r^1\neq\eta_r^2  \quad r\in[t_{\rm fj}(0),t])
\]
On the other hand if $t_{\rm fj}(0)>t$, then $T(t)=0$ thus we get
\[
\Prob(\tilde{J}_t^{id}\in A|t_{\rm fj}(0)>t)=\delta_{J_0}
\]
\end{proof}

\begin{proof}[Proof of Lemma \ref{fin_lem}]
By assumption for every $P_{\lambda_i}$-continuity set $A\subseteq\mathbb{R}^2\times S$
\[
\Prob\left(\left(\frac{\eta_t^i}{\sqrt{t}},\epsilon_t^i\right)^{id}\in A\right)\to P_{\lambda_i}(A)\quad i=1,2
\]
(Recall that $\lambda_i=\lambda$ for $i=1$ and $\lambda_i=\sqrt{1-\lambda^2}$ for $i=2$) where the $P_{\lambda_i}$-s are determined by Corollary $\ref{CLT}$, thus all we have to show is
\begin{equation}
\Prob\left(\frac{\tilde{J}_{T(t)+t}^{id}-\tilde{J}_{T(t)}^{id}}{\sqrt{t}}\in .\bigg|\epsilon_{T(t)}^i=\epsilon_0^i\right)\Rightarrow(\times_{i=1,2} P_{\lambda_i})(.)
\end{equation}
Although this could be veryfied directly, we choose a different approach. For ordinary random walks, the authors - generalizing a result of Bolthausen - established the desired result in a functional context (cf. Corollary 1 in and remark (4) in \cite{PGySz2010b}). For economicity, we omit the proof of an invariance theorem of the RWwIS (although it is not by any means harder than the invariance principle for the ordinary RW) and the obvious generalization of the cited result.
\end{proof}

\section{Remarks}
\begin{enumerate}
\item The aformentioned direct proof of Lemma \ref{fin_lem} is based on $T(t)/t\stackrel{P}{\to} 0$. It suffices to show that
\[
\Prob\left(\left|\frac{\eta_{T(t)+t}^i-\eta_{T(t)}^i-\eta_t^i}{\sqrt{t}}\right|>\epsilon\bigg|\epsilon_{T(t)}^i=\epsilon_0^i\right)\to 0\quad i=1,2
\]
Similarly as in Lemma \ref{place_of_last_coll}, $\eta_{T(t)}^i/\sqrt{t}\stackrel{P}{\to}$, so we can drop it from the above formula using the triangle inequality. Then, we can argue that $T(t)$ being small implies that the difference of $|\eta_{T(t)+t}-\eta_t|$ is small enough such that it converges to zero  in probability in the scaling limit.
\item For ordinary random walk there is a variant of the local theorem which is a better spatial estimate (cf. P7.10 in \cite{S76}). The corresponding theorem for RWwIS is
\[
\sum_{x\in\integers^d}|x|^2\left|h_{t,x}(A|\xi_0=(0,u_0))-\frac{\rho(A)}{(\lambda t)^{d/2}}g_{\sigma}\left(\frac{x}{\sqrt{\lambda t}}\right)\right|=o(1)
\]
This gives the limit of the mean and the variance of the absolute value.
\item For treating the deterministic model, the realistic alternative is to rely upon the averaging method of \cite{ChD09}. Indeed, between two collisions of the disks there typically occur long collision sequences  of the particular disks with the periodic configuration of fixed scatterers. During these long intervals, their orbits become approximately Brownian and their velocities and the normal of impact incoming into a particular collision of the two disks correspond to an equilibrium distribution and finally their outgoing velocities from the collision can be calculated analogously to the collision operator appearing in the derivation of Boltzmann's equation for a hard disk fluid. We plan to return to the deterministic model in the future.

\end{enumerate}

\appendix
\section{Proof of STRP and STMRP results}\label{Proof_RPdMC}
\subsection{Preparatory facts}
To prove our results we need the so called Abelian-Tauberian theorems (see \cite{F70} XIII.5). Before presenting them, introduce the Laplace transform of $F$:
\[
\varphi(z)=\int_0^{\infty}e^{-zx}dF(x)=z\int_0^{\infty}e^{-zx}F(x)dx\qquad z\geq 0
\]
where the second equality can be obtained by partial integration. Clearly by the scaling relation, $\varphi_{\lambda}(z)=\varphi(z/\lambda)$.
\begin{Fact}[Feller]\label{Taub1}
Let H be a measure on $\reals^+$, $\kappa(z)=\int e^{-zx}dH$ the Laplace transform wrt it and $H(x)\equiv H([0,x])$! Then for $\rho\geq 0$,
\[
\frac{\kappa(t/x)}{\kappa(1/x)}\to t^{-\rho}\qquad x\to\infty
\]
and
\[
\frac{H(tx)}{H(x)}\to t^{\rho}\qquad x\to\infty
\]
imply each other, moreover in this case
\begin{equation}\label{taub_1_cons}
\kappa(1/x)\sim H(x)\Gamma(\rho+1)\qquad x\to\infty
\end{equation}
\end{Fact}
A popular reformulation of this result is
\begin{Fact}\label{Taub2}
If $L$ is slowly varying in infinity and $0\leq\rho<\infty$, then
\[
\kappa(1/x)\sim x^{\rho}L(x)\qquad x\to\infty
\]
and
\[
H(x)\sim\frac{1}{\Gamma(\rho+1)}x^{\rho}L(x)\qquad x\to\infty
\]
implies each other.
\end{Fact}
We also need a technical result (cf. \cite{F70} VIII.9 Theorem 1)
\begin{Fact}
If $K(x)$ is a slowly varying function at infinity and $\beta>-1$ then
\[
\int_1^xy^{\beta}K(y)dy=\left(\frac{1}{1+\beta}+o(1)\right)x^{\beta+1}K(x)
\]
\end{Fact}

\subsection{MRP results}
As mentioned in Section 2.3.2, we will investigate the asymptotic behavior of $\Phi_{t,\lambda_0}(A)$ as $t\to\infty$. By conditioning on the first return we can write
\begin{align*}
\Phi_{t,\lambda_0}(A)=&\mathbb{1}_{\{\lambda_0\in A\}}(1-F_{\lambda_0}(t))+\\
&+\int_0^t\int_a^bg(\lambda_0,\lambda_1)\Phi_{t-s,\lambda_1}(A) d\lambda_1dF_{\lambda_0}(s)
\end{align*}
This is, however, not the usual renewal equation, thus we have to generalize the standard results of renewal theory.

The unique solution among the functions which are bounded on bounded intervals is
\begin{equation}\label{solution}
\Phi_{t,\lambda_0}(A)=\int_{A\times[0,t]}(1-F_{\lambda}(t-s))U_{\lambda_0}(ds,d\lambda)
\end{equation}
where
\[
U_{\lambda_0}(t,A)=\EXP\left(N_{t,\lambda_0}\mathbb{1}_{\{\Lambda_{N_{t,\lambda_0}-1}\in A\}}\right)
\]
Note $U_{\lambda_0}(t)=U_{\lambda_0}(t,[a,b])$ is the usual renewal function.

Introduce the measure
\begin{equation}\label{Laplace_U}
\tilde{\omega}_{\lambda_0}(z,A)=\int_{[0,\infty]\times A}e^{-zs}dU_{\lambda_0}(s,\lambda)
\end{equation}
Also let $\tilde{\Xi}(z,A)$ be the Laplace transform of $\Phi_{t,\lambda_0}(A)$ in the variable $t$.

Then by \eqref{solution}, Fubini's theorem, and the product rule of the Laplace transform,
\begin{equation}\label{Laplace_solution}
\tilde{\Xi}_{\lambda_0}(z,A)=\int_A\phi_{\lambda}(z)d\tilde{\omega}_{\lambda_0}(z,\lambda)=\frac{1}{z}\underbrace{\int_A(1-\varphi_1(z/\lambda))\tilde{\omega}(z,d\lambda)}_{W_{\lambda_0}(z,A)}
\end{equation}
where the integration is wrt the measure defined in \eqref{Laplace_U} and
\begin{equation}\label{Laplace_1-F}
\phi_{\lambda}(z)=\int_0^{\infty}e^{-zx}(1-F_{\lambda}(x))dx=\frac{1-\varphi(\frac{z}{\lambda})}{z}
\end{equation}
\begin{proof}[Proof \textit{of Theorem \ref{limstatdistr}}]
We need to show that $W_{\lambda_0}(z,A)\sim\rho_s(A)$ as $z\to 0$ because then the proof is ready by Fact \ref{Taub2} since it implies
\begin{equation}\label{utso}
\int_0^t\Phi_{s,\lambda_0}(A)ds\sim t\rho_s(A)\qquad t\to\infty
\end{equation}
From this, it is easy to see that $\Phi_{t,\lambda_0}(A)\to\rho_s(A)$.

By the monotonicity of $\varphi$,
\begin{equation*}
\left(1-\varphi\left(\frac{z}{b}\right)\right)\tilde{\omega}_{\lambda_0}(z,A)\leq W_{\lambda_0}(z,A)\leq \left(1-\varphi\left(\frac{z}{a}\right)\right)\tilde{\omega}_{\lambda_0}(z,A)
\end{equation*}
We prove that the upper bound is $\sim \rho_s(A)$, then the case of the lower bound is trivially the same.

It is not hard to see that
\begin{align*}
\tilde{\omega}_{\lambda_0}(z,A)&=\mathbb{1}_{\{\lambda_0\in A\}}+g(\lambda_0,A)\varphi_1\left(\frac{z}{\lambda_0}\right)+\\
&+\varphi_1\left(\frac{z}{\lambda_0}\right)\int_a^bg(\lambda_0,\lambda_1)\varphi_1\left(\frac{z}{\lambda_1}\right)g(\lambda_1,A)d\lambda_1+...
\end{align*}
Again by the monotonicity of $\varphi$,
\begin{equation}\label{bound}
\tilde{\omega}_{\lambda_0}(z,A)\leq\mathbb{1}_{\{\lambda_0\in A\}}+\sum_{n=1}^{\infty}\varphi^n\left(\frac{z}{b}\right)g^n(\lambda_0,A)
\end{equation}
We also know that the convergence of the Markov chain is exponential, i.e.
\[
\sup_{A\in\mathcal{B}([a,b])}|g^n(\lambda_0,A)-\rho_s(A)|\leq C_1e^{-C_2n}\qquad C_1,C_2>0
\]
Using this, \eqref{bound} equals
\begin{equation}\label{mindjartvege}
\rho_s(A)\frac{1}{1-\varphi\left(\frac{z}{b}\right)}+RT
\end{equation}
where the remainder term can be estimated
\begin{align*}
|RT|&\leq|\mathbb{1}_{\{\lambda_0\in A\}}-\rho_s(A)|+\sum_{n=1}^{\infty}\varphi_1^n\left(\frac{z}{b}\right)|g^n(\lambda_0,A)-\rho_s(A)|\leq\\
&\leq 2+C_1\frac{1}{1-e^{-C_2}\varphi\left(\frac{z}{b}\right)}=\mathcal{O}(1)\qquad z\to 0
\end{align*}
By multiplying \eqref{mindjartvege} with $1-\varphi(z/a)$ and using that $1-\varphi$ is a slowly varying function as $z\to 0$, the proof is finished.
\end{proof}

\subsection{STMRP with slow tail return times}\label{RenThSlwTail}

\begin{lemma}
\[
\phi_{\lambda}(z)\sim L_{\lambda}\left(\frac{1}{z}\right)\qquad\textrm{so also}\qquad 1-\varphi_{\lambda}(z)\sim L_{\lambda}\left(\frac{1}{z}\right)\qquad z\to 0
\]
where $L_{\lambda}(t)=1-F_{\lambda}(t)$ as $t\to\infty$.
\end{lemma}
\begin{proof}
Let
\[
H_{\lambda}([0,x])\equiv H_{\lambda}(x)=\int_0^x1-F_{\lambda}(s)ds=(1+o(1))x(1-F_{\lambda}(x))+o(x)
\]
by Fact 3, so
\begin{equation}\label{H_reg_var}
\frac{H_1(tx)}{H_1(x)}\to t\qquad x\to\infty
\end{equation}
Due to Fact \ref{Taub1},
\[
\frac{\phi_{\lambda}(tz)}{\phi_{\lambda}(z)}\to\frac{1}{t}\qquad z\to 0
\]
and by \eqref{taub_1_cons},
\[
1-\varphi_{\lambda}(z)\sim \frac{H_{\lambda}(1/z)}{1/z}=z\int_0^{1/z}(1-F_{\lambda}(s))ds\qquad z\to\infty
\]
By \eqref{H_reg_var} and the definition of $H_{\lambda}$, we can write
\[
H_{\lambda}(t)\sim tL_{\lambda}(t)\qquad t\to\infty
\]
where
\[
L_{\lambda}(t)=\frac{1}{t}\int_0^t(1-F_{\lambda}(s))ds\sim 1-F_{\lambda}(t)\qquad t\to\infty
\]
which has $L_{\lambda}(t)=L(\lambda t)=1-F(\lambda t)$. By Fact \ref{Taub2},
\begin{equation}\label{phi_asym}
\phi_{\lambda}(z)\sim\frac{1}{z}L_{\lambda}\left(\frac{1}{z}\right) \qquad 1-\varphi_{\lambda}(z)\sim L_{\lambda}\left(\frac{1}{z}\right) \qquad z\to 0
\end{equation}
\end{proof}
Introduce now
\[
U_{\Lambda}(t)=\sum_{n=0}^{\infty}(\Pi^*)_{i=1}^nF_{\lambda_i}(t)
\]
where $\Pi^*$ denotes the convolution product.
\begin{lemma}
\[
U_{\Lambda}(t)(1-F_1(t))\to 1\qquad t\to\infty
\]
\end{lemma}

\begin{proof}
Consider first the case when all $\lambda_i=\tilde{\lambda}$ and introduce the function $U_{\tilde{\lambda}}^0(t)=\sum_{n=0}^{\infty}F_{\tilde{\lambda}}^{*n}(t)$ and its Laplace transform $\omega_{\tilde{\lambda}}(z)$! Then
\[
\omega_{\tilde{\lambda}}(z)=\sum_{n=0}^{\infty}\varphi_{\tilde{\lambda}}^n(z)=\frac{1}{1-\varphi_{\tilde{\lambda}}(z)}\sim\frac{1}{z\int_0^{1/z}(1-F_{\tilde{\lambda}}(s))ds}\qquad z\to 0
\]
Using this with \eqref{phi_asym}, we have
\[
\omega_{\tilde{\lambda}}(z)\sim\frac{1}{L_{\tilde{\lambda}}(1/z)}\qquad z\to 0
\]
which is a slowly varying function at zero (which means $\omega_{\tilde{\lambda}}(1/z)$ is slowly varying at infinity. By Fact \ref{Taub2},
\[
U_{\tilde{\lambda}}^0(t)\sim \frac{1}{L_{\tilde{\lambda}}(t)}\sim \frac{1}{1-F_{\tilde{\lambda}}(t)}\qquad t\to\infty
\]
so we finally have by $1-F$ being slowly varying that
\begin{equation}\label{U_F_dual}
(1-F(t))U_{\tilde{\lambda}}^0(t)\to 1\qquad t\to\infty
\end{equation}

Now we will prove that \eqref{U_F_dual} also holds in general. Clearly
\[
U_{a}^0(t)\leq U_{\Lambda}(t)\leq U_b^0(t)
\]
so by \eqref{U_F_dual}, we have
\[
U_{\Lambda}(t)(1-F(t))\to 1\qquad t\to\infty
\]
\end{proof}
\begin{proof}[Proof \textit{of Theorem \ref{thm:slowly_var}}.]
Now pick arbitrarily small $\epsilon$ and calculate
\[
\mathbf{P}\left(\frac{Y_{t}}{t}<1-\epsilon\right)
\]
which is
\[
\mathbf{P}\left(\cup_{n=0}^{\infty}\cup_{y\in[\epsilon,1]}\{S_{n}=ty\}\cap\{X_{\lambda_{n+1}}>t(1-y)\}\right)=
\]
\[
=\sum_{n=0}^{\infty}\int_{\epsilon}^{1}\left(1-F_{\lambda_{n+1}}\left(t(1-y)\right)\right)d\left((\Pi^*)_{i=1}^nF_{\lambda_i}(ty)\right)=
\]
\[
=\sum_{n=0}^{\infty}\int_{\epsilon}^{1}\left(1-F\left(\lambda_{n+1}t(1-y)\right)\right)d\left((\Pi^*)_{i=1}^nF_{\lambda_i}(ty)\right)
\]
which has an upper bound by $\lambda_i\geq a>0$:
\[
\int_{\epsilon}^{1}\left(1-F_1(at(1-y))\right)dU_{\Lambda}(ty)\sim\int_{\epsilon}^{1}\frac{1-F_1(at(1-y))}{1-F_1(t)}\frac{dU_{\Lambda}(ty)}{U_{\Lambda}(t)}
\]
as $t\to\infty$. Here the first term goes to zero everywhere except $y=1$, since $1-F(t)$ is slowly varying, while we will show that the measure wrt we are integrating, converges uniformly to the point mass concentrated on zero. To see this, Laplace transform the measure $dU_{\Lambda}(ty)/U_{\Lambda}(t)$ to get
\[
\frac{1}{U_{\Lambda}(t)}\int_0^{\infty}e^{-zy}dU_{\Lambda}(ty)=\frac{\omega_{\Lambda}(z/t)}{U_{\Lambda}(t)}\to 1
\]
due to Theorem \ref{Taub1} which is the Laplace transform of the point mass on zero. Therefore, on the domain of interest this measure uniformly goes to zero.

Since $\epsilon$ is arbitrary, the desired result follows. The result for the residual lifetime can be obtained similarly.
\end{proof}

\end{document}